\documentclass[letterpaper]{article}

\usepackage{graphicx}
\usepackage[nolinenums]{custom_arxiv}
\usepackage{amsthm}
\usepackage{amsmath}
\usepackage{amssymb}
\usepackage{thmtools}
\usepackage{graphicx}
\usepackage{bm}

\usepackage{enumitem}
\usepackage{thm-restate}
\usepackage{lipsum}
\usepackage{tikz}
\usepackage{mleftright}
\usepackage{hyperref}
\usepackage[capitalise,noabbrev]{cleveref}
\usepackage{booktabs}
\usepackage[scientific-notation=true]{siunitx}

\usepackage{microtype} 

\newcommand{\emptyseq}{\varnothing}
\newcommand{\infos}[1]{\mathcal{I}_{#1}}
\newcommand{\seqf}[1]{Q_{#1}}
\newcommand{\seqs}[1]{\Sigma_{#1}}
\newcommand{\conn}{\rightleftharpoons}
\newcommand{\plset}{\{1, \dots, n\}}
\newcommand{\xileaf}[2]{\xi_i(#1; #2)}

\renewcommand{\vec}[1]{\bm{#1}}
\newcommand{\mat}[1]{\bm{#1}}
\newcommand{\symp}[1]{\Delta^{\!#1}}

\newcommand{\placepayoffspacefigurehere}{%
\begin{figure*}[ht]%
\includegraphics[scale=.77]{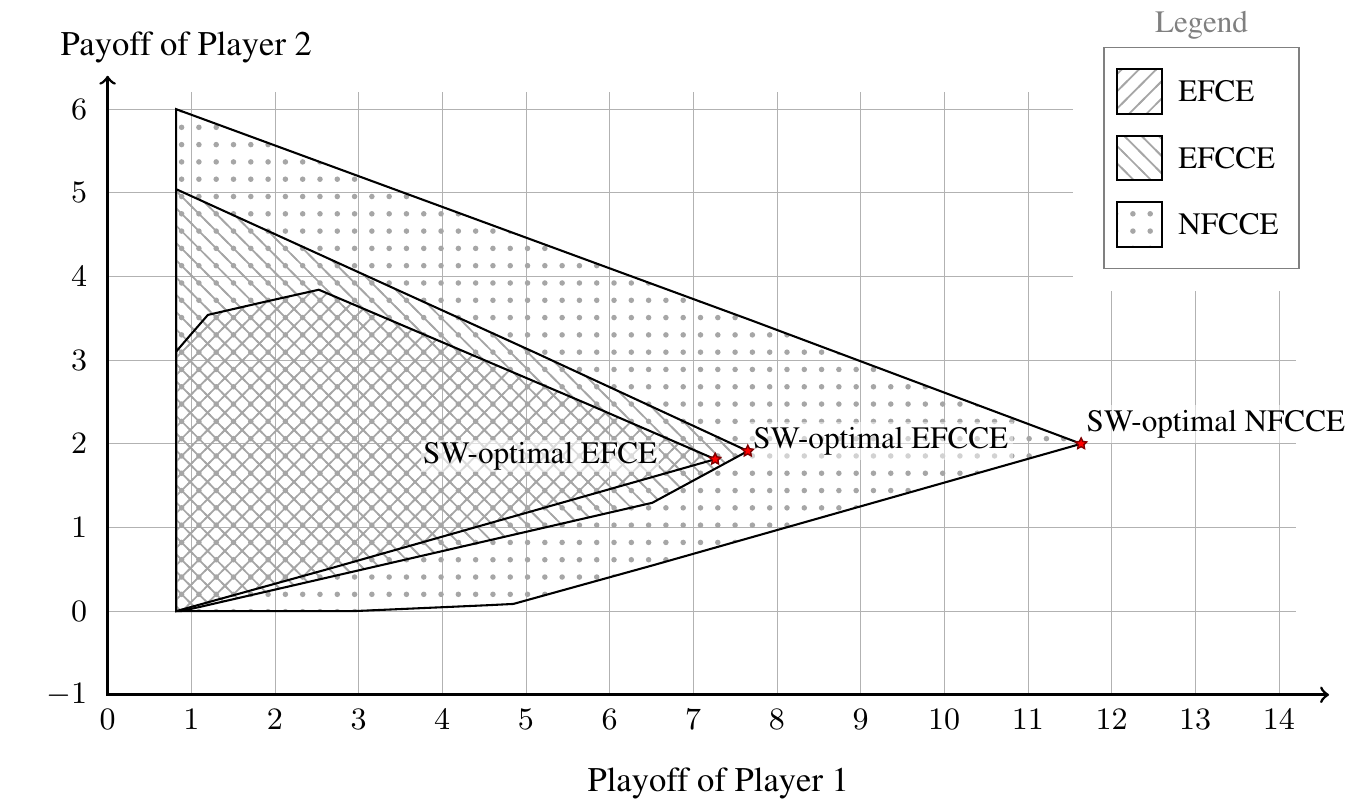}%
\hfill%
\includegraphics[scale=.77]{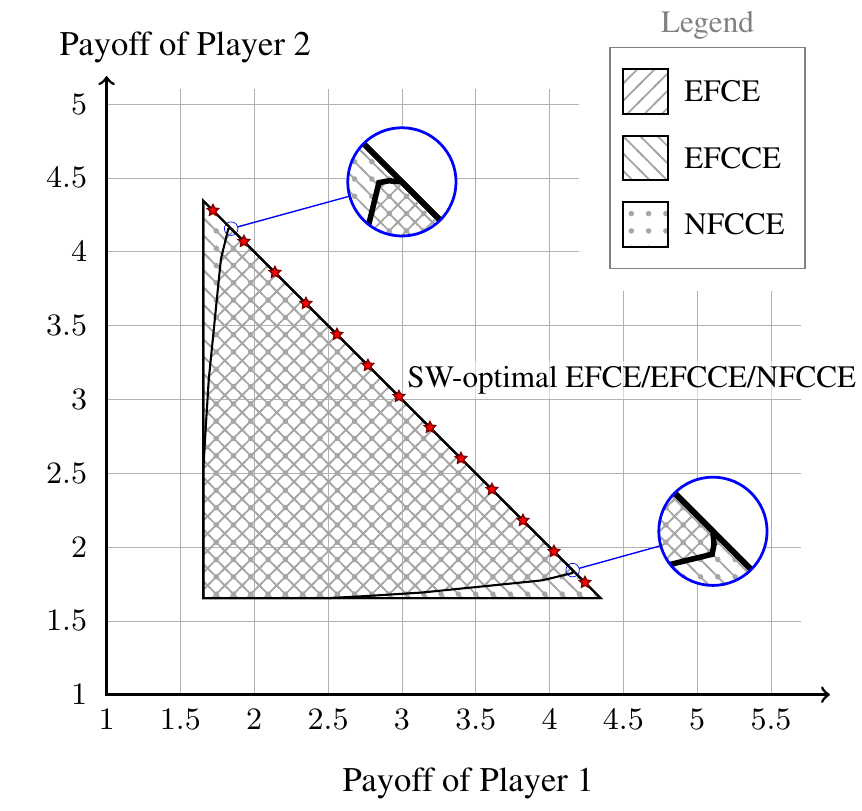}%
    \caption{%
        Space of payoff vectors that can be induced by EFCE, EFCCE and %
        NFCCE in an instance of the Sheriff game (left) and of $3$-card %
        Goofspiel (right). The `SW-optimal' symbols indicate payoffs corresponding to %
        social-welfare-maximizing equilibria.%
    }%
    \label{fig:sheriff goofspiel payoffs}%
\end{figure*}%
}
\newcommand{\placeexperimentshere}{%
\begin{table*}[ht]%
    \centering%
    \includegraphics[scale=.82]{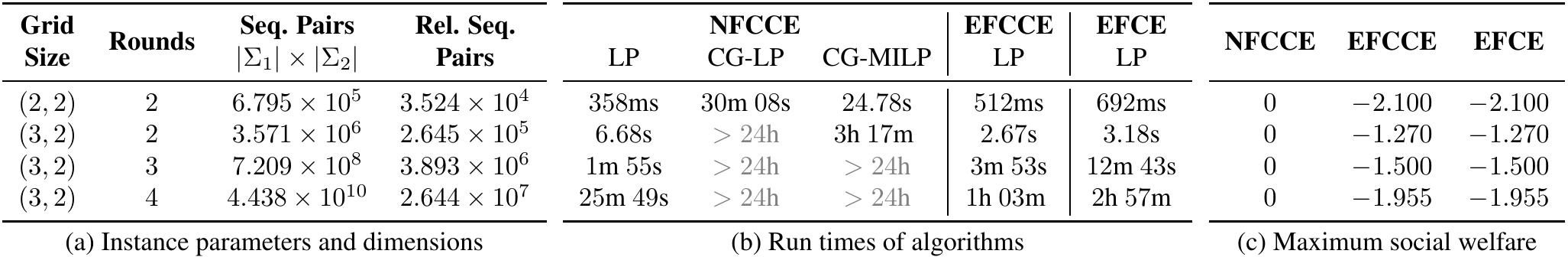}%
    \caption{Experimental results on several instances of the Battleship game.}%
    \label{table:battleship}%
\end{table*}%
\begin{table*}[ht]%
    \centering%
    \includegraphics[scale=.82]{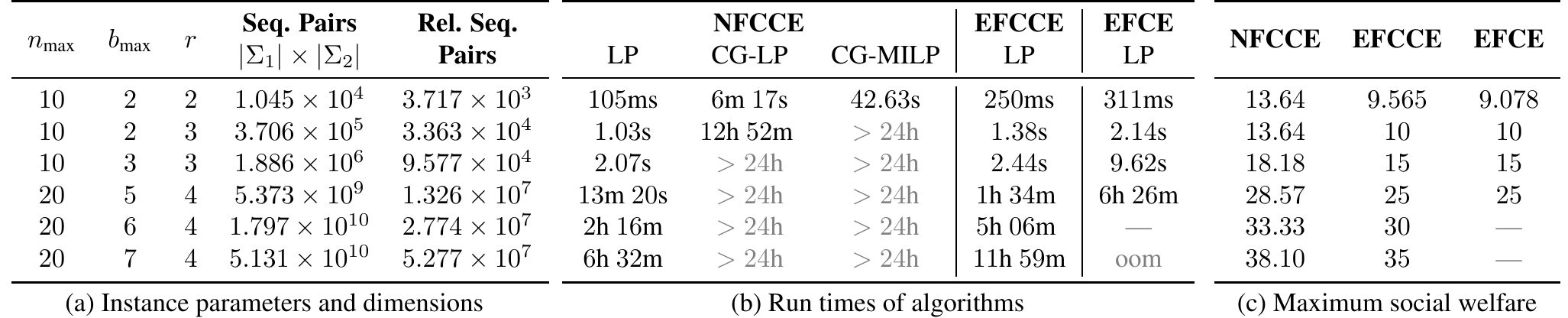}%
    \caption{Experimental results on several instances of the Sheriff game. `oom' means `Out of memory'.}%
    \label{table:sheriff}%
\end{table*}%
}

\title{Coarse Correlation in Extensive-Form Games}
\author{
  Gabriele Farina\\
  Computer Science Department\\
  Carnegie Mellon University\\
  Pittsburgh, PA 15213 \\
  \And
  Tommaso Bianchi\\
  DEIB\\
  Politecnico di Milano\\
  20133 Milan, Italy\\
  \And
  Tuomas Sandholm\\
  Computer Science Department\\
  Carnegie Mellon University\\
  Pittsburgh, PA 15213\\[2mm]
  Strategic Machine, Inc.\\
  Strategy Robot, Inc.\\
  Optimized Markets, Inc.
}

\usepackage{etoolbox}
\ifcsmacro{assumption}{}{
	
}
\ifcsmacro{corollary}{}{
	\newtheorem{corollary}{Corollary}[]
}
\ifcsmacro{definition}{}{
	\newtheorem{definition}{Definition}[]
}
\ifcsmacro{example}{}{
	
}
\ifcsmacro{fact}{}{
	
}
\ifcsmacro{informal_theorem}{}{
	
}
\ifcsmacro{lemma}{}{
	
}
\ifcsmacro{proposition}{}{
	\newtheorem{proposition}{Proposition}[]
}
\ifcsmacro{remark}{}{
	
}
\ifcsmacro{theorem}{}{
	
}

\usepackage{mathtools}
\mathtoolsset{centercolon}
\newcommand{\defeq}{\mathrel{:\mkern-0.25mu=}}

\newcommand{\bbR}{\ensuremath{\mathbb{R}}}

\newcommand{\cX}{\ensuremath{\mathcal{X}}}
\newcommand{\cY}{\ensuremath{\mathcal{Y}}}

\newcommand{\be}{\begin{eqnarray}}
\newcommand{\ee}[1]{\label{#1}\end{eqnarray}}

	\newcommand{\ese}{\end{eqnarray*}}
	\newcommand{\bse}{\begin{eqnarray*}}
	\def\beq{\begin{equation}}
	\def\eeq{\end{equation}}

	\def\fnote#1{\footnote}

	\def\*{{{\LARGE\bf $^*$}}}

	\def\cX{{\cal X}}
	\def\cY{{\cal Y}}

	\def\argmin{\mathop{\rm argmin}}

\begin{document}

    \twocolumn[
        \maketitle
    ]

    \renewcommand\thefootnote{}
    \footnotetext{Author emails: G. Farina \texttt{<gfarina@cs.cmu.edu>}; T. Bianchi \texttt{<tommaso4.bianchi@mail.polimi.it>}; T. Sandholm \texttt{<sandholm@cs.cmu.edu>}. This work was done while T. Bianchi was visiting Carnegie Mellon University.}

    \begin{abstract}
    Coarse correlation models strategic interactions of rational agents complemented by a \emph{correlation device}, that is a mediator that can recommend behavior but not enforce it. Despite being a classical concept in the theory of \emph{normal-form} games for more than forty years, not much is known about the merits of coarse correlation in \emph{extensive-form} settings. In this paper, we consider two instantiations of the idea of coarse correlation in extensive-form games: normal-form coarse-correlated equilibrium (NFCCE), already defined in the literature, and \emph{extensive-form coarse-correlated equilibrium} (EFCCE), which we introduce for the first time. We show that EFCCE is a subset of NFCCE and a superset of the related extensive-form correlated equilibrium. We also show that, in two-player extensive-form games, social-welfare-maximizing EFCCEs and NFCEEs are bilinear saddle points, and give new efficient algorithms for the special case of games with no chance moves. In our experiments, our proposed algorithm for NFCCE is two to four orders of magnitude faster than the prior state of the art.
\end{abstract} 
    \section{Introduction}

As a generic term, \emph{correlated equilibrium} denotes a family of
solution concepts whereby a mediator that can recommend behavior but not enforce it complements the interaction of fully rational agents. Before the game starts, a mediator---also called a \emph{correlation device}---samples a tuple of normal-form plans (one for each player) from a publicly known correlated distribution. He or she then proceeds to privately
ask to each player whether they would like to commit to playing according to the plan that was sampled for them. Being an \emph{equilibrium}, the correlated distribution must be such that no player has benefit in not following the recommendations, assuming all other players follow.
As argued by~\citet{Ashlagi08:Value}, correlated equilibrium is a good candidate to model strategic interactions in which intermediate forms of centralized control can be achieved.

In the context of extensive-form (that is, sequential) games, two different instantiations of the idea of correlated equilibrium are well-known in the literature: \emph{normal-form correlated equilibrium}~(NFCE)~\citep{Aumann74:Subjectivity,Gilboa89:Nash} and \emph{extensive-form correlated equilibrium}~(EFCE)~\citep{Stengel08:Extensive}. The two solution concepts differ in what the mediator reveals to the players. In an NFCE, the mediator privately reveals to each player, just before the game starts, the (whole) normal-form plan that was sampled for them. Players are then free to either play according to the plan, or play any other strategy that they desire. In an EFCE the mediator does not reveal the whole plan to the players before the game starts. Instead, he or she incrementally reveals the plan by recommending individual \emph{moves}. Each recommended move is only revealed when the player reaches the decision point for which the recommendation is relevant. Each player is free to play a move different than the recommended one, but doing so comes at the cost of future recommendations, as the mediator will immediately stop issuing recommendations to players that defect. Because of this deterrent, and because players have to decide whether to follow recommendations knowing less about the sampled normal-form plan than in NFCE, a social-welfare-maximizing EFCE always achieves social welfare equal or higher than any NFCE.

\emph{Coarse} correlated equilibrium differs from correlated equilibrium in that players must decide whether or not to commit to playing according to the recommendations of the mediator \emph{before} observing such recommendations. \emph{Normal-form coarse-correlated equilibrium (NFCCE)}~\citep{Moulin78:Strategically} is the coarse equivalent of NFCE. Before the game starts, players decide whether to commit to playing according to the normal-form plan that was sampled by the mediator (from some correlated distribution known to players), without observing such a plan first. Players who decide to commit will privately receive the plan that was sampled for them; players that decide to not commit will not receive any recommended plan, and are free to play according to any strategy they desire. Since players know less at the time of commitment than either NFCE or EFCE, a social-welfare-maximizing NFCCE is always guaranteed to achieve equal or higher social welfare than any NFCE or EFCE. No coarse equivalent of EFCE is currently known in the literature.

In this paper, we introduce the coarse equivalent of EFCE, which we coin \emph{extensive-form coarse-correlated equilibrium} (EFCCE). It is an intermediate solution concept between EFCE and NFCCE.
Specifically, EFCCE is akin to EFCE in that each recommended move is only revealed when the players reach the decision point for which the recommendation is relevant. However, unlike EFCE, the acting player must choose whether or not to commit to the recommended move \emph{before} such a move is revealed to them, instead of after.
\cref{fig:equilibria} shows how EFCCE fits inside of the family of correlated and coarse-correlated solution concepts.

We prove that EFCCE is always a subset of NFCCE and a superset of EFCE, and give an example of a game in which the three solution concepts lead to distinct solution sets. We also show that the problem of computing a social-welfare-maximizing EFCCE can be represented as a bilinear saddle-point, which can be solved in polynomial time in two-player extensive-form games with no chance moves but not in games with more than two players or two-player games with chance moves. Finally, we note that in two-player games with no chance moves, EFCCE leads to a linear program whose size is smaller than EFCE; because of this, EFCCE can also be used as a computationally lighter relaxation of EFCE---for example, as a routine in the algorithm by~\citet{Cermak16:Using,Bosansky17:Computation} for computing strong Stackelberg equilibrium.

We also show that the problem of computing a social-welfare-maximizing NFCCE can be expressed as a bilinear saddle-point problem, which can be solved in polynomial time in two-player extensive-form games with no chance moves (the problem is known to be NP-hard in games with more than two players and/or chance moves). This formulation is significant, as it enables several new classes of algorithms to be employed to compute social-welfare-maximizing NFCCE. In particular, we show that it enables a linear programming formulation that in our experiments is two to four orders of magnitude faster than the prior state of the art.

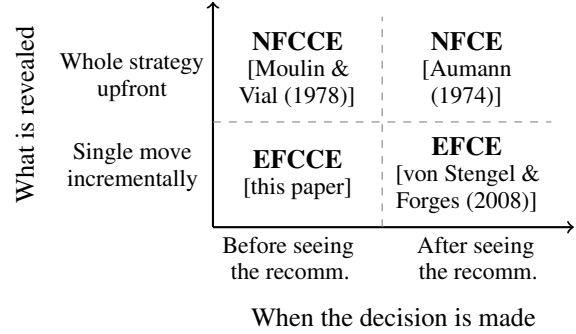
\begin{figure}[ht]
    \begin{tikzpicture}[scale=2]
      \draw[thick,->] (0, 0) --node[xshift=-2.5cm,rotate=90]{What is revealed} (0, 1.5);
      \draw[thick,->] (0, 0) --node[yshift=-1.2cm]{When the decision is made} (2.4, 0);
      \draw[gray,dashed, thin] (0, .7) -- (2.25, .7);
      \draw[gray,dashed, thin] (1.125, 0) -- (1.125, 1.4);

      \node[anchor=north] at (.5, 0) {\begin{minipage}{2cm}\centering\small Before seeing\\the recomm.\end{minipage}};
      \node[anchor=north] at (1.75, 0) {\begin{minipage}{2cm}\centering\small After seeing\\the recomm.\end{minipage}};

      \node[anchor=east] at (0, .4) {\begin{minipage}{1.9cm}\small\centering Single move\\incrementally\end{minipage}};
      \node[anchor=east] at (0, 1) {\begin{minipage}{1.9cm}\small\centering Whole strategy\\upfront\end{minipage}};

      \node[] at (0.5625, .35) {\begin{minipage}{2cm}
                                  \centering \textbf{EFCCE}\\
                                  \small{[this paper]}
                                \end{minipage}};
      \node[] at (1.6875, .35) {\begin{minipage}{1.9cm}
                                  \centering \textbf{EFCE}\\
                                  \small{[\citet{Stengel08:Extensive}]}
                                \end{minipage}};
      \node[] at (0.5625, 1.05) {\begin{minipage}{1.6cm}
                                  \centering \textbf{NFCCE}\\
                                  \small{[\citet{Moulin78:Strategically}]}
                                \end{minipage}};
      \node[] at (1.6875, 1.05) {\begin{minipage}{1.5cm}
                                  \centering \textbf{NFCE}\\
                                  \small{[\citet{Aumann74:Subjectivity}]}
                                \end{minipage}};
    \end{tikzpicture} 
    \caption{Taxonomy of correlated and coarse-correlated equilibria.}
    \label{fig:equilibria}
\end{figure}
    \section{Preliminaries}
\subsection{Extensive-Form Games}
    Extensive-form games are played on a game tree, and can capture both
    sequential and simultaneous moves, as well as private information.
    Each node $v$ in the game tree belongs to exactly one player
    $i \in \plset \cup \{\textsf{c}\}$. Player $\textsf{c}$ is a
    special player called the \emph{chance player}; it is used to denote
    random events that happen in the game, such as drawing a card from a deck
    or tossing a coin. The edges leaving $v$ represent actions that the
    player can take at that node; we denote the set of actions available at
    $v$ as $A_v$.
    In order to capture private information, the set of nodes
    that belong to each player $i \in \plset$ are
    partitioned into a collection $\infos{i}$ of nonempty sets. Each
    $I \in \infos{i}$ is called an \emph{information set} of Player $i$, and
    is a set of nodes that Player $i$ cannot distinguish between, given what
    the player has observed so far. In this paper, we only consider games
    with \emph{perfect recall}, that is games where no player forgets what he
    or she knew earlier. Necessarily, for any $I \in \infos{i}$ and
    $u,v\in I$, it must be $A_u = A_v$, or otherwise Player $i$ would be able
    to distinguish between $u$ and $v$. For this reason, we will often write
    $A_I$ to mean the set of available
    actions at any node in $I$, defined as $A_I \defeq A_u$ for any $u \in I$.
    Finally, two information sets $I_i, I_j$ for Player $i$ and $j$, respectively, are said to be \emph{connected}, denoted $I_i \conn I_j$ if there exist $u \in I_i, v \in I_j$ such that the path from the root to $u$ passes through $v$ or vice versa.

Nodes $v$ for which $A_v$ is empty are called
    \emph{leaves}, and denote an end state of the game. We denote the set of
    leaves of the game with the symbol $Z$. Each $z \in Z$ is associated with
    a tuple of $n$ payoffs (one for each non-chance player); we denote $u_i(z)$ the
    payoff for Player $i \in \plset$ at $z$.

\subsection{Sequences ($\seqs{}$)}
    The set of \emph{sequences} of Player $i$, denoted $\seqs{i}$, is defined
    as the set
    $\seqs{i} \defeq \{(I, a): I \in \infos{i}, a \in A_I\} \cup \{\emptyseq_i\}$,
    where the special sequence $\emptyseq_i$ is called \emph{empty sequence}.
    Given a node $v$ the belongs to Player 1, the \emph{parent sequence} of
    $v$, denoted $\sigma_i(v)$, is defined as the last sequence
    $(I, a) \in \seqs{i}$ encountered on the path from the root to $v$; if no
    such sequence exists (i.e., Player $i$ never acts before $v$), we let
    $\sigma_i(v) = \emptyseq_i$. The \emph{parent sequence} $\sigma(I)$ of an
    information set $I \in \infos{i}$ is defined as
    $\sigma(I) \defeq \sigma(v)$ where $v$ is any node in $I$ (all choices
    produce the same parent sequence, since the game is assumed to have
    perfect recall).
    Finally, we introduce the concept of \emph{relevant} pairs of sequences.
    Given two sequences $\sigma_i$ and $\sigma_j$ for two distinct Player $i$
    and $j$, respectively, we say that the pair $(\sigma_i, \sigma_j)$ is
    relevant if either one sequence is the empty sequence, or if
    $\sigma_i = (I_i, a_i), \sigma_j = (I_j, a_j)$ and $I_i \conn I_j$.

\subsection{Reduced-Normal-Form Plans ($\Pi$)}
    A \emph{normal-form plan} for Player $i$ defines a choice of action $a_I \in A_I$ for \emph{every} information set $I \in \infos{i}$ of the player. However, this representation contains irrelevant information, as some information sets of Player $i$ may become unreachable after the player makes certain decisions higher up the tree. A \emph{reduced-normal-form plan} $\pi$ is a normal-form plans where this irrelevant information is removed: it defines a choice of action $\pi(I) = a_I \in A_I$ for \emph{every} information set $I \in \infos{i}$ that is still reachable as a result of the other choices in $\pi$ itself. We denote the set of reduced-normal-form plans of Player $i$ as $\Pi_i$. Given a sequence $\sigma = (I, a) \in \seqs{i}$, we denote with $\Pi_i(\sigma)$ the (sub)set of reduced-normal-form plans $\pi$ that prescribe that Player $i$ play all actions on the path from the root to any node $v \in I$, including playing action $a$ at $v$. More formally, we let $\Pi_i(\emptyseq_i) \defeq \Pi_i$ and recursively for any $\sigma=(I,a)\in\seqs{i}$ we let $\Pi_i(\sigma) \defeq \Pi_i(\sigma_i(I)) \cap \{\pi \in \Pi_i : \pi(I) = a\}$. We will also make frequent use of the shorthand $\Pi_i(z) \defeq \Pi_i(\sigma_i(z))$ to denote the set of reduced-normal-form plans that allow Player $i$ to reach leaf $z\in Z$, and $\Pi_i(I) \defeq \Pi_i(\sigma_i(I))$ to denote the set of reduced-normal-form plans that allow Player $i$ to reach information set $I \in \infos{i}$.
Finally, a \emph{reduced-normal-form strategy} for Player $i$ is a probability distribution over $\Pi_i$.

\subsection{Polytope of Sequence-Form Strategies ($\seqf{}$)}

    The sequence-form
    representation~\citep{Romanovskii62:Reduction,Koller96:Efficient,Stengel96:Efficient}
    is a more compact way of representing normal-form strategies of a player
    in a perfect-recall extensive-form game. Formally, fix a player
    $i\in\plset$, and let $\mu$ be a reduced-normal-form strategy for Player
    $i$, that is some probability distribution over $\Pi_i$. The
    \emph{sequence-form strategy induced by $\mu$} is the nonnegative real
    vector $\vec{y}$, indexed over $\sigma\in\seqs{i}$, defined as
    \begin{equation}\label{eq:sf definition}
        y(\sigma) \defeq \sum_{\pi \in \Pi_i(\sigma)} \mu(\pi).
    \end{equation}
    The set of sequence-form strategies that can be induced as $\mu$ varies over the set of all possible probability distributions over $\Pi_i$ is denoted $\seqf{i}$. In particular, \citet{Koller96:Efficient} prove that it is a convex polytope (called the \emph{sequence-form polytope}) in $\bbR_+^{|\seqs{i}|}$, described by the affine constraint
\begin{equation*}
    \seqf{i} = \{\vec{y} \in \bbR_+^{|\seqs{i}|}: \mat{F}_i \vec{y} = \mat{f}_i\},
\end{equation*}
where $\mat{F}_i$ is a sparse $|\infos{i}|\times |\seqs{i}|$ matrix with entries in $\{0,1,-1\}$, and $\vec{f}_i$ is a vector with entries in $\{0, 1\}$.

\subsection{Polytope of Extensive-Form Correlation Plans ($\Xi$)}
    Given any probability distribution $\mu$ over
    $\bigtimes_{i=1}^n \Pi_i$ in an extensive-form game, the
    \emph{correlation plan} $\vec{\xi}$
    induced by $\mu$ is defined as the real vector, indexed over tuples
    $(\sigma_1, \dots, \sigma_n) \in \bigtimes_{i=1}^n \seqs{i}$ of
    pairwise-relevant sequences, where each entry is
    \begin{equation}\label{eq:xi definition}
      \xi(\sigma_1, \dots, \sigma_n) \defeq \sum_{\substack{\pi_1 \in \Pi_1(\sigma_1)\\\cdots\\\pi_n \in \Pi_n(\sigma_n)}} \mu(\pi_1, \dots, \pi_n).
    \end{equation}
    The set of correlation plans $\vec{\xi}$ that can be induced as $\mu$
    varies over the set of all possible probability distributions is denoted as $\Xi$ and called the
    \emph{polytope of extensive-form correlation plans}. It is always a
    polytope in a space of dimension polynomial in the input game
    description. Furthermore, in two-player games without chance moves, $\Xi$
    can be described as the intersection of a polynomial number (in the game
    tree) linear constraints, as shown by~\citet{Stengel08:Extensive}. The
    same authors also prove that this property does not always hold in games
    with more than two players and/or chance moves.

    Finally, for any $i \in \plset$, $\sigma \in \seqs{i}$, and $z \in Z$, we introduce the following notation that we will use often in the remainder of this paper:
    \begin{equation*}
        \xileaf{\sigma}{z} \defeq \xi(\sigma_1(z), \dots, \sigma_{i-1}(z), \sigma, \sigma_{i+1}(z), \dots, \sigma_n(z)).
    \end{equation*} 
    \section{Saddle-Point Formulation of NFCCE}\label{sec:saddle_point_nfcce}

In this section, we show that the problem of computing an NFCCE that achieves social welfare at least $\tau$ (for some given $\tau \in \bbR$) in an $n$-player extensive-form game with perfect recall can
be expressed as a bilinear saddle-point problem, that is as an optimization problem of the form
\[
  \argmin_{\vec{x} \in \cX} \max_{\vec{y} \in \cY}\ \vec{x}^{\!\top}\!\! \mat{A} \vec{y},
\]
where $\cX$ and $\cY$ are convex and compact sets. In our specific case, $\cX$ and $\cY$ will be convex polytopes in low-dimensional spaces (in particular, $\cX = \Xi$ and $\cY \subseteq \seqf{1}\times\dots\times \seqf{n}\times \symp{n+1}$).
As we will show later, this formulation immediately implies that in two-player games with no chance moves, a social-welfare-maximizing NFCCE can be computed in polynomial time as the solution of a linear program.

We now go through the steps that enable us to formulate the problem of computing an NFCCE as a bilinear saddle-point problem. The general structure of the argument is similar to that of~\citet{Farina19:Correlation} in the context of EFCE, and we will use it again later when dealing with EFCCE.

By definition, a correlated distribution $\mu$ over $\bigtimes_{i=1}^n \Pi_i$ is an NFCCE if no player has an incentive to unilaterally deviate from the recommended plan assuming that nobody else does. More formally, let $i$ be any player, and let $\hat\mu_i$ be any probability distribution over $\Pi_i$, independent of $\mu$. Playing according to $\hat\mu$ must give Player $i$ an expected utility $\hat u_i$ at most equal to the expected utility $u_i$ of committing to the mediator's recommendation. In order to express $\hat u_i$ and $u_i$ as a function of $\mu$ and $\hat\mu$, it is necessary to quantify the probability of the game ending in any leaf $z \in Z$. When the Player deviates and plays according to $\hat\mu_i$, the probability that the game ends in $z$ is equal to the probability that the mediator samples a plan $\pi_j \in \Pi_j(z)$ for any Player $j$ other than $i$, and that Player $i$ samples a plan $\pi_i \in \Pi_i(z)$. Correspondingly, using the independence of $\mu$ and $\hat\mu_i$, we can write
\begin{align*}
\hat u_i \!=\! \sum_{z\in Z}\! \mleft[ u_i(z) \!\mleft(\!\!\!\sum_{\substack{\pi_i \in \Pi_i~~\\~~\pi_{j} \in \Pi_{j}(z)\ \forall j\neq i}} \hspace{-7mm}\mu(\pi_1, \dots, \pi_n)\!\!\mright)\!\!
\mleft(\sum_{\pi_i\in\Pi_i(z)} \!\!\!\!{\hat\mu}_{i}(\pi_i)\!\!\mright) \!\!\mright]\!.
\end{align*}

On the other hand, the probability that leaf $z$ is reached when all players commit to the mediator's recommendation is equal to the probability that the mediator samples from $\mu$ plans $\pi_j \in \Pi_j(z)$ for all players $j \in \{1,\dots,n\}$:
\begin{align}\label{eq:utility nfcce}
u_i = \sum_{z\in Z} \mleft[u_i(z) \mleft(\sum_{\pi_{j} \in \Pi_{j}(z)\ \forall j} \hspace{-2mm}\mu(\pi_1, \dots, \pi_n) \mright)\!\mright]\!.
\end{align}

Using the definition of extensive-form correlation plan~\eqref{eq:xi definition} and sequence-form strategy~\eqref{eq:sf definition} we can convert the requirement that $\hat u_i \le u_i$ for all choices of $i$ and $\hat\mu_i$ into the following equivalent condition:

\begin{proposition}\label{prop:nfcce}
    An extensive-form correlation plan $\vec{\xi} \in \Xi$ is an NFCCE if and
    only if the following inequality holds for any player $i\in\{1,\dots, n\}$ and sequence-form strategy $\vec{y}_i \in \seqf{i}$:
    \begin{align}\label{eq:xi incentive nfcce}
        \sum_{z\in Z} u_i(z) \xileaf{\emptyseq_i}{z} y_i(\sigma_i(z)) \le \sum_{z\in Z} u_i(z) \xileaf{\sigma_i(z)}{z}.
    \end{align}
\end{proposition}

Inequality~\eqref{eq:xi incentive nfcce} is of the form $\vec{\xi}^{\!\top}\!\! \mat{A}_i \vec{y}_i - \vec{b}_i^{\!\top}\! \vec{\xi} \le 0$, where $\mat{A}_i$ and $\vec{b}_i$ are suitable sparse matrices/vectors that only depend on $i$. With this new notation, we can rewrite the condition in \cref{prop:nfcce} as follows: $\vec{\xi} \in \Xi$ is an NFCCE if and only if
\begin{align}
    &\ \ \ \max_{i=1}^n \max_{\vec{y}_i \in \seqf{i}} \mleft\{ \vec{\xi}^{\!\top}\!\! \mat{A}_i \vec{y}_i - \vec{b}_i^{\!\top}\! \vec{\xi} \mright\} \le 0 \nonumber\\
\iff\ \ & \max_{
    \substack{\vec{\lambda}\in\symp{n}\\
    \vec{y}_i \in \seqf{i} \ \forall i}}
        \mleft\{  \sum_{i=1}^n \lambda_i\mleft(\vec{\xi}^{\!\top}\!\! \mat{A}_i \vec{y}_i - \vec{b}_i^{\!\top}\! \vec{\xi} \mright)\mright\} \le 0 \nonumber\\
\iff\ \ & \max_{
    \substack{\vec{\lambda}\in\symp{n}\\
    \tilde{\vec{y}}_i \in \lambda_i\seqf{i} \ \forall i}}
        \mleft\{ \sum_{i=1}^n \vec{\xi}^{\!\top}\!\! \mat{A}_i \tilde{\vec{y}}_i - \lambda_i \vec{b}_i^{\!\top}\! \vec{\xi} \mright\} \le 0 \label{eq:primal nfcce},
\end{align}
where in the last transformation we operated a change of variable $\tilde{\vec{y}_i} \defeq \lambda_i \vec{y}_i$; it is a simple exercise to prove that this change of variable is legitimate and that the domain of the maximization is a convex polytope.
Since an NFCCE always exists, in particular any $\vec{\xi}$ such that
\begin{equation}\label{eq:nfcce saddle point}
  \vec{\xi} \in \argmin_{\vec{\xi}\,\in\,\Xi}\max_{
    \substack{\vec{\lambda}\in\symp{n}\\
    \tilde{\vec{y}}_i \in \lambda_i\seqf{i} \ \forall i}}
        \mleft\{ \sum_{i=1}^n \vec{\xi}^{\!\top}\!\! \mat{A}_i \tilde{\vec{y}}_i - \lambda_i \vec{b}_i^{\!\top}\! \vec{\xi} \mright\}
\end{equation}
is guaranteed to be an NFCCE. Since the domains of the minimization and maximization problems are both convex polytopes, and since the objective function is bilinear, the optimization problem in~\eqref{eq:nfcce saddle point} is a bilinear saddle-point problem.

\subsection{Enforcing a Lower Bound on Social Welfare}
Given an NFCCE $\mu$, social welfare is defined as $\text{SW} \defeq \sum_{i=1}^n u_i$, where $u_i$ is as in Equation~\eqref{eq:utility nfcce}. Hence, it is a linear function of the correlation plan $\vec{\xi}$, which can be expressed as $\text{SW}: \Xi \ni \vec{\xi} \mapsto \vec{c}^{\!\top}\!\vec{\xi}$ where $\vec{c} \defeq \sum_{i=1}^n \vec{b}_i$. Consequently, an NFCCE that guarantees a given lower bound $\tau$ on the social welfare can be expressed as in~\eqref{eq:nfcce saddle point} where the domain of the minimization is changed from $\vec{\xi}\in\Xi$ to $\vec{\xi} \in \Xi \cap \{\vec{\xi}: \vec{c}^{\!\top}\! \vec{\xi} \ge \tau\}$. Note that this preserves the polyhedral nature of the optimization domain.


Finally, the same construction can be used verbatim if social welfare is replaced with any linear function of $\vec{\xi}$.

\subsection{Connection to Linear Programming}
The saddle-point formulation in~\eqref{eq:nfcce saddle point} can be mechanically translated into a linear program by taking the dual of the internal maximization problem, that is of~\ref{eq:primal nfcce}. Specifically, the dual problem is the linear program
\refstepcounter{equation}
\label{lp:max dev nfcce}
\begin{equation*}
   \thetag\theequation : \mleft\{\begin{array}{rll}
        \min & u \\
        \text{s.t.}     & u - \vec{v}_i^{\!\top}\!\vec{f}_i + \vec{b}_i^{\!\top}\! \vec{\xi} \ge 0 & \forall\, i \in \{ 1,\dots,n \} \\[.5mm]
                        & \mat{F}_i^{\!\top}\! \vec{v}_i - \mat{A}_i^{\!\top}\! \vec{\xi} \ge \vec{0} & \forall\, i \in \{ 1,\dots,n \}\\[2mm]
                        & u \in \bbR, \vec{v}_i \in \bbR^{|\infos{i}|} & \forall\, i \in \plset.\\
    \end{array}\mright.
\end{equation*}

(See the Preliminaries section for the meaning of $\mat{F}_i$ and $\vec{f}_i$). Problem~\eqref{lp:max dev nfcce} has a polynomial number of variables and constraints.

By strong duality, the value of~\eqref{lp:max dev nfcce} is the same as the value of the primal problem, that is the maximum `deviation benefit' $\hat u_i - u_i$ across all players $i\in \plset$ and probability distributions $\hat \mu_i$ over $\Pi_i$. Hence, we can find an NFCCE $\vec{\xi}$ that maximizes any given objective $\vec{c}^{\!\top}\! \vec{\xi}$ by adding the constraint $u \le 0$ and solving the modified linear program
\refstepcounter{equation}
\label{lp:nfcce}
\begin{equation*}
    \thetag\theequation : \mleft\{\begin{array}{rll}
        \max & \vec{c}^{\!\top}\!\vec{\xi} \\
        \text{s.t.}     & u - \vec{v}_i^{\!\top}\!\vec{f}_i + \vec{b}_i^{\!\top}\! \vec{\xi} \ge 0 & \forall\, i \in \{ 1,\dots,n \} \\[.5mm]
                        & \mat{F}_i^{\!\top}\! \vec{v}_i - \mat{A}_i^{\!\top}\! \vec{\xi} \ge \vec{0} & \forall\, i \in \{ 1,\dots,n \}\\[.5mm]
                        & u \le 0\\[2mm]
                        & \vec{\xi} \in \Xi\\
                        & u \in \bbR, \vec{v}_i \in \bbR^{|\infos{i}|} & \forall\, i \in \plset.\\
    \end{array}\mright.
\end{equation*}
The linear program above always has a polynomial number of variables, but potentially an exponential number of constraints because of the condition $\vec{\xi} \in \Xi$. However, in two-player extensive-form games with no chance moves,~\eqref{lp:nfcce} is \emph{guaranteed} to have a polynomial number of constraints as $\Xi$ can be described compactly~\citep{Stengel08:Extensive}. Hence, in those games a social-welfare-maximizing NFCCE can be computed in polynomial time after setting $\vec{c} \defeq \sum_{i=1}^n \vec{b}_i$. 
    \section{EFCCE: An Intermediate Solution Concept}

In this section, we introduce a new solution concept which we coin
\emph{extensive-form coarse-correlated equilibrium (EFCCE)}. It combines the
idea of \emph{coarse correlation}---that is, players must decide whether they
want to commit to following the recommendations issued by the correlation
device, before observing such recommendations---with the idea of
\emph{extensive-form correlation}---that is, recommendations are revealed
incrementally as the players progress on the game tree. Specifically, EFCCE
is akin to EFCE in that each recommended move is only revealed when the
players reach the decision point for which the recommendation is relevant.
However, unlike EFCE, the acting player must choose whether or not to commit
to the recommended move \emph{before} such a move is revealed to them,
instead of after. Each choice is binding only with respect to the decision
point for which the choice is made, and players can make different choices at
different decision points. Just like EFCE, defections (that is, deciding to
\emph{not} commit to following the correlation device's recommended move)
come at the cost of future recommendations, as the correlation device will
stop issuing recommendations to the defecting player. As with all correlated
equilibria, the correlated distribution from which the recommendations are
sampled must be such no player has incentives to unilaterally deviate when no
other player does.

\subsection{Saddle-Point Formulation}\label{sec:saddle_point_efcce}

In this section, we show that an EFCCE can also be expressed as the solution to a bilinear saddle-point problem. To do so, we use the idea of \emph{trigger agents}~\citep{Gordon08:No,Dudik09:Sampling}:

\begin{definition}
    Let $i \in \plset$ be a player, let $\hat I \in \infos{i}$ be an information set for Player $i$, and let $\hat \mu$ be a probability distribution over $\Pi_i(\hat I)$. An \emph{$(\hat I, \hat \mu)$-trigger agent for Player $i$} is a player that commits to and follows all recommendations issued by the mediator until they reach a node $v \in \hat I$ (if any). When any node $v \in \hat I$ is reached, the player `gets triggered', stops committing the recommendations and instead plays according to a reduced-normal-form plan sampled from $\hat \mu$ until the game ends.
\end{definition}

By definition, a correlated distribution $\mu$ over $\bigtimes_{i=1}^n \Pi_i$ is an EFCCE when, for all $i \in \{1,\dots, n\}$, the value $u_i$ that Player $i$ obtains by following the recommendations is at least as large as the expected utility $\hat u_{\hat I}$ attained by any $(\hat I, \hat \mu)$-trigger agent for that player (assuming nobody else deviates).
The expected utility for Player $i$ when everybody commits to following the mediator's recommendations is as in Equation~\eqref{eq:utility nfcce}.
In order to express the expected utility of the $(\hat I, \hat \mu)$-trigger agent, we start by computing the probability of the game ending in each possible leaf $z\in Z$. Let $(\pi_1, \dots, \pi_n)$ be the tuple of reduced-normal-form plans that was sampled by the mediator. Two cases must be distinguished:
\begin{itemize}[nolistsep,itemsep=1mm]
  \item The path from the root to $z$ passes through a node $v \in \hat I$. We denote the set of such leaves as $Z_{\hat I}$. In this case, the trigger agent commits to following all recommendations until just before $\hat I$, and then plays according to a reduced-normal-form plan $\hat \pi \in \Pi_i(\hat I)$ sampled from the distribution $\hat \mu$ from $\hat I$ onwards. Hence, the following conditions are necessary and sufficient for the game to terminate at $z$: $\pi_j \in \Pi_j(z)$ for all $j \in \{1,\dots, n\} \setminus \{i\}$, $\pi_i \in \Pi_i(\hat I)$, and $\hat \pi \in \Pi_i(z)$. Correspondingly, the probability that the game ends at $z \in Z_{\hat I}$ is
      \begin{equation}\label{eq:Z I prob}
        p_{z} \defeq \!\mleft(\sum_{\substack{\pi_i \in \Pi_i(\hat I)\\\pi_{j} \in \Pi_{j}(z)\ \forall j\neq i}} \hspace{-5mm}\mu(\pi_1, \dots, \pi_{n})\!\mright)\!\!\mleft(\sum_{\pi_i\in\Pi_i(z)}\!\!\! {\hat\mu}(\pi_i)\!\mright)\!.
      \end{equation}

  \item Otherwise, the trigger agent never gets triggered, and instead commits to following all recommended moves until the end of the game. The probability that the game ends at $z \in Z \setminus Z_{\hat I}$ is therefore
      \begin{equation}\label{eq:Z not I prob}
        q_z \defeq \sum_{\pi_{j} \in \Pi_{j}(z)\ \forall j} \hspace{-2mm}\mu(\pi_1, \dots, \pi_n).
      \end{equation}
\end{itemize}
With this information, the expected utility of the $(\hat I, \hat \mu)$-trigger agent is computed as
\begin{align*}
    \hat u_{\hat I} = \sum_{z \in Z_{\hat I}} u_i(z)\, p_z\ + \sum_{z \in Z \setminus Z_{\hat I}} \!\! u_i(z)\, q_z.
\end{align*}

Using to the definition of extensive-form correlation plan~\eqref{eq:xi definition} and sequence-form strategy~\eqref{eq:sf definition}, we can rewrite the condition $u_i \ge \hat u_{\hat I}$ (which must hold for all choices of $i$, $\hat I \in \infos{i}$ and probability distribution $\hat \mu$ over $\Pi_i(\hat I)$) compactly as in the following proposition.

\begin{proposition}\label{prop:efcce}
    An extensive-form correlation plan $\vec{\xi} \in \Xi$ is an EFCCE if and
    only if the following inequality holds for any player $i\in\{1,\dots, n\}$, information set $\hat{I} \in \infos{i}$, and sequence-form strategy $\vec{y}_{i, \hat I} \in \seqf{i}$ such that $y_{i, \hat I}(\sigma(\hat I)) = 1$:
    \begin{align}\label{eq:xi incentive efcce}
        & \sum_{z\in Z_{\hat I}} u_i(z) \xileaf{\sigma_i(\hat{I})}{z} y_{i, \hat I}(\sigma_i(z)) \nonumber \\[-3mm]
        & \hspace{2.9cm} \le \sum_{z\in Z_{\hat I}} u_i(z) \xileaf{\sigma_i(z)}{z}.
    \end{align}
\end{proposition}

Inequality~\eqref{eq:xi incentive efcce} is again in the form $\vec{\xi}^{\!\top}\!\! \mat{A}_{i, \hat I} \vec{y}_{i, \hat I} - \vec{b}_{i, \hat I}^{\!\top} \vec{\xi} \le 0$ where $\mat{A}_{i, \hat I}$ and $\vec{b}_{i, \hat I}$ are suitable matrices/vectors that only depend on the trigger information set $\hat I$ of Player $i$. From here, one can follow the same steps that we already took in the case of NFCCE and obtain a bilinear saddle-point formulation and a linear program for EFCCE. For space reasons, we only state the linear program, which we also implemented and tested (see Experimental Evaluation section).

\refstepcounter{equation}
\label{lp:efcce}
\begin{equation*}
    \thetag\theequation : \mleft\{\!\!\begin{array}{rll}
        \max&\!\!\!\vec{c}^{\!\top}\!\vec{\xi} \\
        \text{s.t.} &\!\!\! u - w_{i, \hat I} - \vec{v}_{i, \hat I}^{\!\top} \vec{f}_i + \vec{b}_{i, \hat I}^{\!\top} \vec{\xi} \ge 0 & \!\!\forall i, \hat I \in \infos{i} \\[.5mm]
                        &\!\!\! \mat{F}_i^{\!\top} \vec{v}_{i, \hat I} + w_{i, \hat I} - \mat{A}_{i, \hat I}^{\!\top} \vec{\xi} \ge \vec{0} & \!\!\forall i, \hat I \in \infos{i} \\[.5mm]
                        &\!\!\! u \le 0\\[2mm]
                        &\!\!\! \vec{\xi} \in \Xi\\
                        &\!\!\! u \in \bbR, w_{i, \hat I} \in \bbR, \vec{v}_{i, \hat I} \in \bbR^{|\infos{i}|} &\!\!\forall i, \hat I \in \infos{i}.
    \end{array}\mright.
\end{equation*}

The linear program~\eqref{lp:efcce} has a polynomial number of variables, and in two-player game with no chance it has also a polynomial number of constraints because of the polynomial description of $\Xi$ \citep{Stengel08:Extensive}. In particular, in two-player games with no chance moves, a social-welfare-maximizing EFCCE can be computed in polynomial time by setting the objective function $\vec{c}^{\!\top}\! \vec{\xi}$ to be the social welfare
\[
  \vec{c}^{\!\top}\! \vec{\xi} \defeq \sum_{z\in Z} \mleft[\mleft(\sum_{i=1}^n u_i(z) \mright) \xi(\sigma_1(z), \dots, \sigma_n(z))\mright].
\]

Finally, we remark that the EFCCE linear program~\eqref{lp:efcce} has more constraints and variables than NFCCE (see \eqref{lp:nfcce}), but less than EFCE (see Supplemental Material). Empirically, this results in intermediate run times compared to NFCCE and EFCE, as confirmed by our experiments below.

\subsection{Complexity Results}

    In this section, we discuss complexity results relating to the problem of computing a social-welfare maximizing EFCCE.

    As we have already pointed out, in the case of two-player games without chance, the linear program in~\eqref{lp:efcce} has a polynomial number of constraints and variables, and can therefore be solved in polynomial time using standard LP technology. We now point out that, as in NFCCE and EFCE, the same does not hold in general for games with more than two players and/or chance moves. In particular, the following results can be easily obtained by using the same reduction employed by~\citet{Stengel08:Extensive}:

    \begin{definition}[$\text{SW}_{\text{EFCCE}}(\kappa)$]
    	Given an extensive-form game $\Gamma$ and a real number $\kappa$, $\text{SW}_\text{EFCCE}(\kappa)$ denotes the problem of deciding whether or not $\Gamma$ admits an EFCCE with social welfare at least $\kappa$.
    \end{definition}

    \begin{restatable}{proposition}{prophardnesstwoplchance}\label{prop:hardness 2 players chance}
        $\text{SW}_{\text{EFCCE}}(\kappa)$ is NP-Hard in two-player games with chance moves.
    \end{restatable}

    \begin{restatable}{proposition}{prophardnessthreepl}\label{prop:hardness 3 players}
        $\text{SW}_{\text{EFCCE}}(\kappa)$ is NP-Hard in three-player games, with or without chance moves.
    \end{restatable} 
    \section{Relationships Between Equilibria}

\placepayoffspacefigurehere 

In this section, we analyze some relationships between EFCE, EFCCE and
NFCCE. We start with the following inclusion lemma, which shows that the
solution concept that we just introduced, EFCCE, is a superset of EFCE and
a subset of NFCCE (a proof of all propositions is available in the Supplemental Material):

\begin{restatable}{proposition}{propinclusionequilibria}\label{prop:inclusion of equilibria}
    Let $\Gamma$ be a perfect-recall extensive-form game. Then we have the
    following inclusion of equilibria
    \[
      \text{EFCE} \subseteq \text{EFCCE} \subseteq \text{NFCCE}.
    \]
\end{restatable}

\cref{prop:inclusion of equilibria} applies to games with more
than two players and/or chance moves as well.
Let $U_\text{EFCE},
U_\text{EFCCE}, U_\text{NFCCE}$ denote the set of expected payoff vectors
that can be induced by EFCE, EFCCE and NFCCE, respectively. Then, \cref{prop:inclusion of equilibria} is an important ingredient for our next proposition:

\begin{restatable}{proposition}{propinclusionpayoffs}\label{prop:inclusion of payoffs}
  The sets $U_\text{EFCE}, U_\text{EFCCE}, U_\text{NFCCE}$ are convex
  polytopes. Furthermore,
  $U_\text{EFCE} \subseteq U_\text{EFCCE} \subseteq U_\text{NFCCE}$.
\end{restatable}

\cref{prop:inclusion of equilibria} also implies the following relationship between the maximum social welfare that can be obtained by EFCE, EFCCE and NFCCE:

\begin{corollary}\label{cor:sw inequalities}
    Let
    $\text{SW}^*_\text{EFCE},\text{SW}^*_\text{EFCCE},\text{SW}^*_\text{NFCCE}$
    denote the maximum social welfare that can be reached by EFCE, EFCCE and
    NFCCE, respectively. Then, one has the inequality
    $
        \text{SW}^*_\text{EFCE} \le \text{SW}^*_\text{EFCCE} \le \text{SW}^*_\text{NFCCE}.
    $
\end{corollary}

\cref{fig:sheriff goofspiel payoffs} shows the set of payoff vectors that can
be induced by EFCE, EFCCE and NFCCE in an instance of the Sheriff
game~\citep{Farina19:Correlation} (left) and an instance of a 3-card
Goofspiel game~\citep{Ross71:Goofspiel} (right).\footnote{The polytopes of
reachable payoffs were computed with the help of Polymake, a tool for
computational polyhedral
geometry~\citep{Gawrilow00:Polymake,Assarf17:Computing}.} In the Sheriff game
instance, we have that both inclusions in \cref{prop:inclusion of payoffs} are
strict, while in the Goofspiel game only the inclusion $U_\text{EFCE}
\subsetneq U_\text{EFCCE}$ is strict, while
$U_\text{EFCCE} = U_\text{NFCCE}$. The appendix contains an instance of a
Battleship game~\citep{Farina19:Correlation} in which only the inclusion
$U_\text{EFCCE} \subsetneq U_\text{NFCCE}$ is strict, while
$U_\text{EFCE} = U_\text{EFCCE}$.

In the game of \ref{fig:sheriff goofspiel payoffs} (left) the inequalities
between the values of the maximum social welfare (\cref{cor:sw inequalities})
are both strict, whereas in \ref{fig:sheriff goofspiel payoffs} (right) both
inequalities are equalities.

    \section{Experimental Evaluation}

\placeexperimentshere

We experimentally compare NFCCE, EFCCE and EFCE both in terms of maximum social welfare and run time.

In our experiments, we use instances from three different two-player games with no chance moves: Sheriff~\citep{Farina19:Correlation}, Battleship~\citep{Farina19:Correlation} and Goofspiel~\citep{Ross71:Goofspiel}, whose full descriptions are available in the Supplemental Material. Sheriff is a bargaining game, in which two players---respectively, the Smuggler and the Sheriff---must settle on an appropriate bribe so as to avoid the Sheriff inspecting the Smuggler's cargo, which might or might not contain illegal items. Battleship is a parametric version of the classic board game, where two competing fleets take turns at shooting at each other. Finally, Goofspiel is a card game in which two players repeatedly bid to win a common public card, which will be discarded in case of bidding ties. The three games were chosen as to illustrate three different applications in which an intermediate form of centralized control (the correlation device) is beneficial: bargaining in Sheriff, conflict resolution in Battleship, and bidding in Goofspiel.

We used Gurobi 8.1.1~\cite{Gurobi} to solve the linear programs~\eqref{lp:nfcce} for NFCCE,~\eqref{lp:efcce} for EFCCE, and~\eqref{lp:efce} for EFCE (the latter is given in the Supplemental Material). We use the barrier algorithm without crossover, and we let Gurobi automatically the recommended number of threads for execution.
All experiments were run on a 64-core machine with 512 GB of RAM.

Our experimental results are available in \cref{table:battleship} for Battleship, \cref{table:sheriff} for Sheriff and \cref{table:goofspiel} in the Supplemental Material for Goofspiel. Each table is split into three parts. \textbf{Part (a)} contains information about the parameters that were used to generate the game instances (refer to the Supplemental Material for a detailed description of their effects). It also shows the size of the instances in terms of number of sequences pairs, defined as the product $|\Sigma_1| \times |\Sigma_2|$ of the number of sequences of the players, and number of \emph{relevant} pairs of sequences (see Preliminaries section). \textbf{Part (b)} compares the run times of our algorithm (column `LP'). In the case of NFCCE, we also compare against the only known polynomial-time algorithms to compute social-welfare-maximizing NFCCE in extensive-form games, which are both based on the column generation technique, and have been introduced by~\citet{Celli19:Computing}.
In particular, we implemented both the algorithm based on a linear programming oracle, denoted `CG-LP' in the tables, and the `more practical' algorithm `CG-MILP' based on a mixed integer linear programming oracle that is proposed by~\citet{Celli19:Computing} as a practically faster approach for most applications. CG-LP is guaranteed to compute a social-welfare-maximizing NFCCE in polynomial time, whereas CG-MILP requires exponential time in the worst case.
Finally, \textbf{Part (c)} reports the value of the maximum social welfare that can be attained by NFCCE, EFCCE and EFCE.

\subsection{Comparison of Run Time}

As expected, increasing the coarseness of the equilibrium---from EFCE to EFCCE to NFCCE---reduces the linear program size which results in a smaller run time. Empirically, the NFCCE linear program is up to four times faster than the EFCCE linear program, and the EFCCE linear program is in turn between two to four times faster than the EFCE linear program.
Furthermore, our results indicate that the NFCCE linear program that we develop in~\eqref{lp:nfcce} is two to four orders of magnitude faster than CG-LP and CG-MILP, and it is able to scale to game instances even up to five orders of magnitude larger than CG-LP and CG-MILP can in 24 hours. We believe that this difference in performance is partly due to the fact that the algorithms by~\citet{Celli19:Computing} have a number of variables that scales with the total number $|\seqs{1}|\times |\seqs{2}|$ of sequence pairs in the game, whereas our linear programming formulation has a number of columns that grows with the number of \emph{relevant} sequence pairs, which is only a tiny fraction of the total number of sequence pairs in practice.

\subsection{Comparison of Maximum Social Welfare}

Our results experimentally confirm~\cref{cor:sw inequalities}: as the coarseness of the equilibrium increases from EFCE to EFCCE to NFCCE, so does the value of the maximum social welfare that the mediator can induce.
The maximum social welfare attained by NFCCE is strictly larger than EFCCE and EFCE in Battleship (\cref{table:battleship}) and Sheriff (\cref{table:sheriff}), while it is the same in Goofspiel (\ref{table:goofspiel} in the Supplemental Material).

Experimentally, the maximum social welfare that can be obtained through EFCCE is often equal to the maximum social welfare that can be obtained through EFCE. While this does not imply that the set of reachable payoffs is the same (see \cref{fig:sheriff goofspiel payoffs}), it is an indication of the fact that EFCCE is a tight relaxation of EFCE which can be solved up to four times faster than EFCE in practice.

    \section{Conclusions}

In this paper we studied two instantiations of the idea of coarse correlation in extensive-form games: normal-form coarse-correlated equilibrium and extensive-form coarse-correlated equilibrium. For both solution concepts, we gave saddle-point problem formulations and linear programs. 

We proved that EFCCE, which we introduced for the first time, is an intermediate solution concept between NFCCE and the extensive-form correlated equilibrium introduced by~\citet{Stengel08:Extensive}. In particular, the set of payoffs that can be reached by EFCCE is always a superset of those that can be reached by EFCE, and a subset of those that can be reached by NFCCE. 
Empirically, EFCCE is a tight relaxation of EFCE, and a social-welfare-maximizing EFCCE can be computed up to four times faster than EFCE. This suggests that EFCCE could be a suitable and faster alternative in algorithms that rely on EFCE, such as the algorithm by~\citet{Cerny18:Incremental} in the context of Stackelberg games.

Finally, we compared the run time of our algorithm for computing social-welfare-maximizing NFCCEs, and showed that it is two to four orders of magnitude faster than the only previously known algorithms by~\citet{Celli19:Computing}. Our algorithm can also scale to game instances up to five orders of magnitude larger than the prior state of the art, thus enabling the computation of coarse-correlated solution concepts in medium-sized extensive-form games for the first time.

    \section*{Acknowledgments}
    This material is based on work supported by the National Science Foundation under grants IIS-1718457, IIS-1617590, IIS-1901403, and CCF-1733556, and the ARO under award W911NF-17-1-0082. Gabriele Farina is supported by a Facebook Fellowship.

    \bibliographystyle{custom_arxiv}
    \bibliography{dairefs}

\begin{thebibliography}{19}
\providecommand{\natexlab}[1]{#1}
\providecommand{\url}[1]{\texttt{#1}}
\expandafter\ifx\csname urlstyle\endcsname\relax
  \providecommand{\doi}[1]{doi: #1}\else
  \providecommand{\doi}{doi: \begingroup \urlstyle{rm}\Url}\fi

\bibitem[Ashlagi et~al.(2008)Ashlagi, Monderer, and
  Tennenholtz]{Ashlagi08:Value}
Ashlagi, I., Monderer, D., and Tennenholtz, M.
\newblock On the value of correlation.
\newblock \emph{Journal of Artificial Intelligence Research}, 33:\penalty0
  575--613, 2008.

\bibitem[Assarf et~al.(2017)Assarf, Gawrilow, Herr, Joswig, Lorenz, Paffenholz,
  and Rehn]{Assarf17:Computing}
Assarf, B., Gawrilow, E., Herr, K., Joswig, M., Lorenz, B., Paffenholz, A., and
  Rehn, T.
\newblock Computing convex hulls and counting integer points with polymake.
\newblock \emph{Mathematical Programming Computation}, 9\penalty0 (1):\penalty0
  1--38, Mar 2017.
\newblock ISSN 1867-2957.
\newblock \doi{10.1007/s12532-016-0104-z}.

\bibitem[Aumann(1974)]{Aumann74:Subjectivity}
Aumann, R.
\newblock Subjectivity and correlation in randomized strategies.
\newblock \emph{Journal of Mathematical Economics}, 1:\penalty0 67--96, 1974.

\bibitem[Bo{\v{s}}ansk{\`y} et~al.(2017)Bo{\v{s}}ansk{\`y}, Br{\^a}nzei,
  Hansen, Lund, and Miltersen]{Bosansky17:Computation}
Bo{\v{s}}ansk{\`y}, B., Br{\^a}nzei, S., Hansen, K.~A., Lund, T.~B., and
  Miltersen, P.~B.
\newblock Computation of {S}tackelberg equilibria of finite sequential games.
\newblock \emph{ACM Transaction on Economics and Computation (TEAC)},
  5\penalty0 (4):\penalty0 23:1--23:24, December 2017.
\newblock ISSN 2167-8375.

\bibitem[Celli et~al.(2019)Celli, Coniglio, and Gatti]{Celli19:Computing}
Celli, A., Coniglio, S., and Gatti, N.
\newblock Computing optimal ex ante correlated equilibria in two-player
  sequential games.
\newblock In \emph{Proceedings of the 18th International Conference on
  Autonomous Agents and MultiAgent Systems}, pp.\  909--917. International
  Foundation for Autonomous Agents and Multiagent Systems, 2019.

\bibitem[{\v{C}}erm{\'a}k et~al.(2016){\v{C}}erm{\'a}k, Bo{\v{s}}ansk{\`y},
  Durkota, Lis{\'y}, and Kiekintveld]{Cermak16:Using}
{\v{C}}erm{\'a}k, J., Bo{\v{s}}ansk{\`y}, B., Durkota, K., Lis{\'y}, V., and
  Kiekintveld, C.
\newblock Using correlated strategies for computing {Stackelberg} equilibria in
  extensive-form games.
\newblock In \emph{AAAI}, 2016.

\bibitem[{\v{C}}ern{\`y} et~al.(2018){\v{C}}ern{\`y}, Bo{\`y}ansk{\`y}, and
  Kiekintveld]{Cerny18:Incremental}
{\v{C}}ern{\`y}, J., Bo{\`y}ansk{\`y}, B., and Kiekintveld, C.
\newblock Incremental strategy generation for stackelberg equilibria in
  extensive-form games.
\newblock In \emph{Proceedings of the 2018 ACM Conference on Economics and
  Computation}, pp.\  151--168. ACM, 2018.

\bibitem[Dudik \& Gordon(2009)Dudik and Gordon]{Dudik09:Sampling}
Dudik, M. and Gordon, G.~J.
\newblock A sampling-based approach to computing equilibria in succinct
  extensive-form games.
\newblock In \emph{Proceedings of the Twenty-Fifth Conference on Uncertainty in
  Artificial Intelligence}, pp.\  151--160. AUAI Press, 2009.

\bibitem[Farina et~al.(2019)Farina, Ling, Fang, and
  Sandholm]{Farina19:Correlation}
Farina, G., Ling, C.~K., Fang, F., and Sandholm, T.
\newblock Correlation in extensive-form games: Saddle-point formulation and
  benchmarks.
\newblock ArXiV preprint, 2019.

\bibitem[Gawrilow \& Joswig(2000)Gawrilow and Joswig]{Gawrilow00:Polymake}
Gawrilow, E. and Joswig, M.
\newblock \emph{Polymake: a Framework for Analyzing Convex Polytopes}, pp.\
  43--73.
\newblock Birkh{\"a}user Basel, Basel, 2000.
\newblock ISBN 978-3-0348-8438-9.
\newblock \doi{10.1007/978-3-0348-8438-9_2}.

\bibitem[Gilboa \& Zemel(1989)Gilboa and Zemel]{Gilboa89:Nash}
Gilboa, I. and Zemel, E.
\newblock {N}ash and correlated equilibria: Some complexity considerations.
\newblock \emph{Games and Economic Behavior}, 1:\penalty0 80--93, 1989.

\bibitem[Gordon et~al.(2008)Gordon, Greenwald, and Marks]{Gordon08:No}
Gordon, G.~J., Greenwald, A., and Marks, C.
\newblock No-regret learning in convex games.
\newblock In \emph{Proceedings of the 25\textsuperscript{th} international
  conference on Machine learning}, pp.\  360--367. ACM, 2008.

\bibitem[Gurobi~Optimization(2019)]{Gurobi}
Gurobi~Optimization, L.
\newblock Gurobi optimizer reference manual, 2019.
\newblock URL \url{http://www.gurobi.com}.

\bibitem[Koller et~al.(1996)Koller, Megiddo, and {von
  Stengel}]{Koller96:Efficient}
Koller, D., Megiddo, N., and {von Stengel}, B.
\newblock Efficient computation of equilibria for extensive two-person games.
\newblock \emph{Games and Economic Behavior}, 14\penalty0 (2), 1996.

\bibitem[Moulin \& Vial(1978)Moulin and Vial]{Moulin78:Strategically}
Moulin, H. and Vial, J.-P.
\newblock Strategically zero-sum games: The class of games whose completely
  mixed equilibria cannot be improved upon.
\newblock \emph{International Journal of Game Theory}, 7\penalty0
  (3-4):\penalty0 201--221, 1978.

\bibitem[Romanovskii(1962)]{Romanovskii62:Reduction}
Romanovskii, I.
\newblock Reduction of a game with complete memory to a matrix game.
\newblock \emph{Soviet Mathematics}, 3, 1962.

\bibitem[Ross(1971)]{Ross71:Goofspiel}
Ross, S.~M.
\newblock Goofspiel—the game of pure strategy.
\newblock \emph{Journal of Applied Probability}, 8\penalty0 (3):\penalty0
  621--625, 1971.

\bibitem[{von Stengel}(1996)]{Stengel96:Efficient}
{von Stengel}, B.
\newblock Efficient computation of behavior strategies.
\newblock \emph{Games and Economic Behavior}, 14\penalty0 (2):\penalty0
  220--246, 1996.

\bibitem[von Stengel \& Forges(2008)von Stengel and
  Forges]{Stengel08:Extensive}
von Stengel, B. and Forges, F.
\newblock Extensive-form correlated equilibrium: Definition and computational
  complexity.
\newblock \emph{Mathematics of Operations Research}, 33\penalty0 (4):\penalty0
  1002--1022, 2008.

\end{thebibliography}

\iftrue
    \clearpage
    \twocolumn[
        \begin{center}
            \huge\bf Supplemental Material
            \vspace{2cm}
        \end{center}
    ]
    \appendix
    \section{Formulation of EFCE}

In this section, we show that an EFCE can also be expressed as the solution to a bilinear saddle-point problem. To do so, we resort again to the idea of \emph{trigger agents}~\citep{Gordon08:No,Dudik09:Sampling}, slightly modifying the definition of trigger agent that we have given in the section about EFCCE to allow for deviations to happen \emph{after} the recommendations have been received:

\begin{definition}
    Let $i \in \plset$ be a player, let $\hat \sigma = (\hat I, \hat a) \in \seqs{i}$ be a sequence for Player $i$, and let $\hat \mu$ be a probability distribution over $\Pi_i(\hat I)$. An \emph{$(\hat \sigma, \hat \mu)$-trigger agent for Player $i$} is a player that follows all recommendations issued by the mediator unless they get recommended to play $\hat a$ at information set $\hat I$; if this happens, the player `gets triggered', stops following the recommendations and instead plays according to a reduced-normal-form plan sampled from $\hat \mu$ until the game ends.
\end{definition}

By definition, a correlated distribution $\mu$ over $\bigtimes_{i=1}^n \Pi_i$ is an EFCE when, for all $i \in \plset$, the value $u_i$ that Player $i$ obtains by following the recommendations is at least as large as the expected utility $\hat u_{\hat \sigma}$ attained by any $(\hat \sigma, \hat \mu)$-trigger agent for that player (assuming nobody else deviates).
The expected utility for Player $i$ when everybody follows the mediator's recommendations is as in Equation~\eqref{eq:utility nfcce}. On the other hand, the expected utility for the $(\hat \sigma, \hat \mu)$-trigger agent can be computed similarly to the one for the $(\hat I, \hat \mu)$-trigger agent that we have computed in the section of EFCCE. We thus compute the probability of the game ending in any terminal node $z \in Z$, this time distinguishing three cases:
\begin{itemize}[nolistsep,itemsep=1mm]
  \item The path from the root to $z$ includes playing action $\hat a$ at information set $\hat I$. We denote the set of such leaves as $Z_{\hat \sigma}$. In this case, the trigger agent follows all recommendations until $\hat a$ get recommended, and then plays according to a reduced-normal-form plan $\hat \pi \in \Pi_i(\hat I)$ sampled from the distribution $\hat \mu$ from $\hat I$ onwards. To reach leaf $z$, however, we need to have that the reduced-normal-form plan $\hat \pi$ includes playing $\hat a$ at $\hat I$. Hence, the following conditions are necessary and sufficient for the game to terminate at $z$: $\pi_j \in \Pi_j(z)$ for all $j \in \plset \setminus \{i\}$, $\pi_i \in \Pi_i(\hat \sigma)$, and $\hat \pi \in \Pi_i(z)$. Correspondingly, the probability that the game ends at $z \in Z_{\hat \sigma}$ is
      \begin{equation}\label{eq:Z sigma prob efce}
        p_{z} \defeq \mleft(\sum_{\substack{\pi_i \in \Pi_i(\hat \sigma)\\\pi_{j} \in \Pi_{j}(z)\ \forall j\neq i}} \hspace{-5mm}\mu(\pi_1, \dots, \pi_{n})\!\mright)\!\mleft(\sum_{\pi_i\in\Pi_i(z)}\!\!\! {\hat\mu}(\pi_i)\!\mright).
      \end{equation}

  \item The path from the root to $z$ passes through information set $\hat I$ but does not include playing action $\hat a$ at $\hat I$. The game can end in this state both if the mediator recommends all the players to play in order to reach $z$, and thus the trigger agent never gets triggered, or if the mediator recommends to play $\hat a$ at $\hat I$ but then the trigger agent deviates and plays according to the reduced-normal-form plan $\hat \pi \in \Pi_i(z)$ sampled from the distribution $\hat \mu$. Hence, the probability that the game ends at $z \in Z_{\hat I} \setminus Z_{\hat \sigma}$ is the sum of two terms, as follows:
    \begin{align*}\label{eq:Z I not sigma prob efce}
      q_{z} \defeq &\sum_{\pi_{j} \in \Pi_{j}(z)\ \forall j} \hspace{-2mm}\mu(\pi_1, \dots, \pi_n) \\
      + &\mleft(\sum_{\substack{\pi_i \in \Pi_i(\hat \sigma)\\\pi_{j} \in \Pi_{j}(z)\ \forall j\neq i}} \hspace{-5mm}\mu(\pi_1, \dots, \pi_{n})\!\mright)\!\mleft(\sum_{\pi_i\in\Pi_i(z)}\!\!\! {\hat\mu}(\pi_i)\!\mright).
    \end{align*}

  \item Otherwise, the trigger agent never gets triggered, and instead followings all recommended moves until the end of the game. The probability that the game ends at $z \in Z \setminus Z_{\hat I}$ is therefore
      \begin{equation}\label{eq:Z not I prob efce}
        r_z \defeq \sum_{\pi_{j} \in \Pi_{j}(z)\ \forall j} \hspace{-2mm}\mu(\pi_1, \dots, \pi_n).
      \end{equation}
\end{itemize}
With this information, the expected utility of the $(\hat \sigma, \hat \mu)$-trigger agent is computed as
\begin{align*}
    \hat u_{\hat \sigma} = \sum_{z \in Z_{\hat \sigma}} u_i(z)\, p_z + \sum_{z \in Z_{\hat I} \setminus Z_{\hat \sigma}} u_i(z)\, q_z + \sum_{z \in Z \setminus Z_{\hat I}} u_i(z)\, r_z.
\end{align*}

By following the same steps that have already been taken for NFCCE and EFCCE, one can now write the constraint defining an EFCE in the following, compact way:

\begin{align}\label{eq:xi incentive efce}
    &  \sum_{z\in Z_{\hat \sigma}} u_i(z) \xileaf{\hat \sigma}{z} y_{i, \hat \sigma}(\sigma_i(z)) \nonumber\\[-3mm]
    & \hspace{2.6cm} \le \sum_{z\in Z_{\hat \sigma}} u_i(z) \xileaf{\sigma_i(z)}{z},
\end{align}

which needs to hold for all players $i \in \plset$ and sequences $\hat \sigma \in \seqs{i}$. Inequality~\eqref{eq:xi incentive efce} is in the form $\vec{\xi}^{\!\top}\!\! \mat{A}_{i, \hat \sigma} \vec{y}_{i, \hat \sigma} - \vec{b}_{i, \hat \sigma}^{\!\top} \vec{\xi} \le 0$ where $\mat{A}_{i, \hat \sigma}$ and $\vec{b}_{i, \hat \sigma}$ are suitable matrices/vectors that only depend on the trigger sequence $\hat \sigma$ of Player $i$. From this formulation, one can obtain a linear program for computing an EFCE optimizing over any linear function of $\xi$, following the steps that we have outlined in the secion about linear programming for NFCCE. We only give here the final LP.

\refstepcounter{equation}
\label{lp:efce}
\begin{equation*}
    \thetag\theequation :
    \mleft\{\begin{array}{rll}
        \max \!\!\!\! & \vec{c}^{\!\top}\!\vec{\xi} \\
        \text{s.t.} \!\!\!\! & u - \vec{v}_{i, \hat \sigma}^{\!\top} \vec{f}_{i} - w_{\hat \sigma} + \vec{b}_{i, \hat \sigma}^{\!\top} \vec{\xi} \ge 0 & \forall i,  \hat \sigma \! \in \! \seqs{i} \\[.5mm]
                        & \mat{F}_i,^{\!\top} \vec{v}_{i, \hat \sigma} + w_{\hat \sigma} - \mat{A}_{i, \hat \sigma}^{\!\top} \vec{\xi} \ge \vec{0} & \forall i,\, \hat \sigma \! \in \! \seqs{i} \\[.5mm]
                        & u \le 0\\[2mm]
                        & \vec{\xi} \in \Xi\\
                        & u \in \bbR\\
                        & w_{\hat \sigma} \in \bbR,\vec{v}_{i,\, \hat \sigma} \in \bbR^{|\seqs{i}|} & \forall i,\, \hat \sigma \! \in \! \seqs{i}.\\
    \end{array}\mright.
\end{equation*}

    \section{Proofs}

\propinclusionequilibria*
\begin{proof}
    We break the proof into two parts, which can be read independently.

    \vspace{2mm}
    \noindent\textbf{$\text{EFCE} \subseteq \text{EFCCE}$}\quad Let
    $\vec{\xi} \in \Xi$ be an EFCE. We need to show that, given any player
    $i$, decision point $\hat I \in \infos{i}$ and extensive-form strategy
    $\vec{y}_{\hat I} \in \seqf{i}$ such that $y_{\hat I}(\sigma_i(I)) = 1$,
    Inequality~\eqref{eq:xi incentive efcce} is satisfied. The crucial
    ingredient in the proof is the fact that for any $I\in\infos{i}$ and
    $z \in Z$,
    \[
        \xileaf{\sigma_i(I)}{z} = \sum_{\sigma_i \in I} \xileaf{\sigma_i}{z}
    \]
    by definition of $\Xi$~\eqref{eq:xi definition}. Hence,
    \begin{align}
        &\sum_{z\in Z_{\hat I}} u_i(z) \xi_(\sigma_i(\hat I), z) y_{\hat I}(\sigma_i(z)) \nonumber\\
        = &\sum_{z\in Z_{\hat I}}\sum_{\hat\sigma \in \hat I} u_i(z) \xileaf{\hat \sigma}{z} y_{\hat I}(\sigma_i(z)) \nonumber\\
        = &\sum_{\hat\sigma \in \hat I}\mleft[\sum_{z\in Z_{\hat I}} u_i(z) \xileaf{\hat\sigma}{z} y_{\hat I}(\sigma_i(z))\mright]. \tag{$\star$}\label{eq:inclusion proof 1}
    \end{align}
    Since $\vec{\xi}$ is an EFCE and $y_{\hat I}(\sigma_i(\hat I)) = 1$ by hypothesis, the quantity in square brackets in~\eqref{eq:inclusion proof 1} is upper bounded as in Inequality~\eqref{eq:xi incentive efce}. Hence we can write
    \begin{align*}
        \eqref{eq:inclusion proof 1} &\le \sum_{\hat\sigma \in \hat I} \sum_{z\in Z_{\hat\sigma}} u_i(z) \xileaf{\sigma_i(z)}{z}\\
        &= \sum_{z\in Z_{\hat I}} u_i(z) \xileaf{\sigma_i(z)}{z},
    \end{align*}
    where the second equality follows from the fact that the collection of sets $\{Z_{\hat\sigma}\}_{\hat\sigma \in \hat I}$ forms a partition of $Z_{\hat I}$. This shows that Inequality~\eqref{eq:xi incentive efcce} holds for any applicable choice of player $i$, decision point $\hat I$ and deviation strategy $\vec{y}_{\hat I}$, and therefore $\vec{\xi}$ is an EFCCE.

    \vspace{2mm}
    \noindent\textbf{$\text{EFCCE} \subseteq \text{NFCCE}$}\quad Let $\vec{\xi} \in \Xi$
    be an EFCCE. We need to show that, for any
    player $i$ and extensive-form deviation strategy $y_i \in \seqf{i}$,
    Inequality~\eqref{eq:xi incentive nfcce} is satisfied. To this end, let
    $\infos{i}^*$ be the set of \emph{initial} decision points for the
    given player, defined as all information sets that has an empty parent
    sequence; in symbols,
    $\infos{i}^* \defeq \{I \in \infos{i} : \sigma_i(I) = \emptyseq_i\}$.
    The collection of sets $\{Z_{I}\}_{I \in \infos{i}^*}$ is a partition
    of the set of all leaves $Z$. Hence, using the fact that
    $\sigma_i(I) = \emptyseq_i$ for all $I \in \infos{I}^*$:
    \begin{align*}
        &\sum_{z \in Z} u_i(z) \xileaf{\emptyseq_i}{\sigma_i(z}) y_i(\sigma_i(z)) \\
        = &\sum_{I \in \infos{i}^*} \sum_{z\in Z_I} u_i(z) \xileaf{\emptyseq_i}{z} y_i(\sigma_i(z))\\
        = &\sum_{I \in \infos{i}^*} \sum_{z\in Z_I} u_i(z) \xileaf{\sigma_i(I)}{z} y_i(\sigma_i(z))\\
        \le &\sum_{I \in \infos{i}^*} \sum_{z\in Z_I} u_i(z) \xileaf{\sigma_i(z)}{z}\\
        = &\sum_{z \in Z} u_i(z) \xileaf{\sigma_i(z)}{z},
    \end{align*}
    where the inequality follows from
    Inequality~\eqref{eq:xi incentive efcce}, which is applicable since
    $\vec{\xi}$ is an EFCCE by hypothesis, and
    $y(\sigma_i(I)) = y(\emptyseq_i) = 1$ by definition of $\seqf{i}$.
    This shows that Inequality~\eqref{eq:xi incentive nfcce} holds for any
    player $i$, and therefore $\vec{\xi}$ is an NFCCE.
\end{proof}

\propinclusionpayoffs*
\begin{proof}
    First, observe that the set of NFCCE is a convex polytope, since it is the intersection of $\Xi$ with Inequality~\ref{eq:xi incentive nfcce} for all relevant instantiations (i.e., all players $i$ in the case of NFCCE). (Equivalent statements hold for EFCCE and EFCE). Second, the function that maps a $\vec{\xi} \in \Xi$ to the tuple of expected payoffs (one for each player) under the assumptions that players do not deviate from the recommendation strategy encoded by $\vec{\xi}$, namely
\[
  \vec{\xi} \mapsto \begin{pmatrix}\vdots \\[.5mm] \sum_{z \in Z} u_i(z) \xileaf{\sigma_i(z)}{z} \\[-.1mm] \vdots \end{pmatrix}_{\!i \,\in\, \plset},
\]
is linear.
Since the image of a convex polytope with respect to a linear functions is a convex polytope, the first part of the statement follows. The second part of the statement follows trivially from \cref{prop:inclusion of equilibria}.
\end{proof}

\prophardnesstwoplchance*
\begin{proof}
	Given a SAT instance $(C, V)$ in disjunctive normal form, where $C$ is the set of clauses and $V$ the set of variables, we can build an extensive-form game as follows (see \cref{fig:sat_chance_tree} for an example of such a game tree):
	\begin{itemize}
		\item~at the start of the game, chance selects one action in the set $\{a_{\phi} | \phi \in C\}$ uniformly at random, that is it picks non-deterministically a clause in the SAT formula;
		\item~each action $a_{\phi}$ leads to a node $h_{\phi} \in I_{\phi}$ of Player 1, where he can choose an action to play in the set $\{a_{\phi,l} | l \in \phi\}$, that is it selects a literal in the clause $\phi$;
		\item~all actions $a_{\phi,l}$ with $l = v$ or $l = \bar{v}$ for some $v \in V$ lead to node $h_{\phi,l} \in I_v$ of Player 2, where he can choose an action to play in the set $\{a_v, a_{\bar{v}}\}$, that is it selects a truth assignment for variable $v$;
		\item~actions $\{a_v, a_{\bar{v}}\}$ lead to terminal nodes, where the players receive utility $(0,0)$ if $l = v$ and action $a_{\bar{v}}$ was played or if $l = \bar{v}$ and action $a_v$ was played, and they receive utility $(1,1)$ otherwise.
	\end{itemize}
	If the SAT formula is satisfiable, this game admits a pure Nash Equilibrium in which Player 2 plays the truth assignment that satisfies it and Player 1 selects a satisfiable literal for each clause; in this case, the expected utility for both player is $1$ and thus the social welfare is $2$. If the SAT formula is not satisfiable, then for any truth assignment there exist at least one clause for which no literal can evaluate to true, which means that any strategy profile will lead to at least one $(0,0)$ outcome and thus the social welfare will be strictly smaller than $2$. Since any Nash Equilibrium is trivially also an EFCCE, and since $2$ is the optimal outcome of the game, the pure Nash Equilibrium described above is also an EFCCE. Hence, any polynomial time algorithm for deciding whether it exists an EFCCE with social-welfare greater than $2$ could be employed to decide whether a SAT formula is satisfiable or not.
\end{proof}

\prophardnessthreepl*
\begin{proof}
	The proof is similar to the one for two players with chance, with the only difference that we now have to employ the introduction of a third player to simulate the random chance move. This requires only a small modification in the extensive-form game. In fact, it is sufficient to replace the player at the root node with Player 3, that can choose one action in the set $\{a_{\phi} | \phi \in C\}$. The game then proceeds unchanged, with Player 3 receiving as payoff $1-u_1$ where $u_1$ is the payoff received in the same terminal node by Player 1.

	If the SAT formula is satisfiable, then the game admits at least one Nash Equilibrium with a social welfare of $2$. Regardless of the strategy employed by Player 3, Player 1 can in fact always select one literal for each clause such that they all evaluate to true under the truth assignment played by Player 2. If the SAT formula is not satisfyable, then every Nash Equilibrium has a social welfare strictly smaller than $2$, as Player 3 is incentivized to play in a way as to reach as many non-satisfiable clauses, after which no strategy of Player 1 and $2$ can get anything better than a social welfare of $1$. Thus, the same argument of the two player with chance proof still holds, hence the SAT problem can be reduced to the $\text{SW}_{\text{EFCCE}}(\kappa)$ one.
\end{proof}

\begin{figure}
    \centering
    \includegraphics[width=\linewidth]{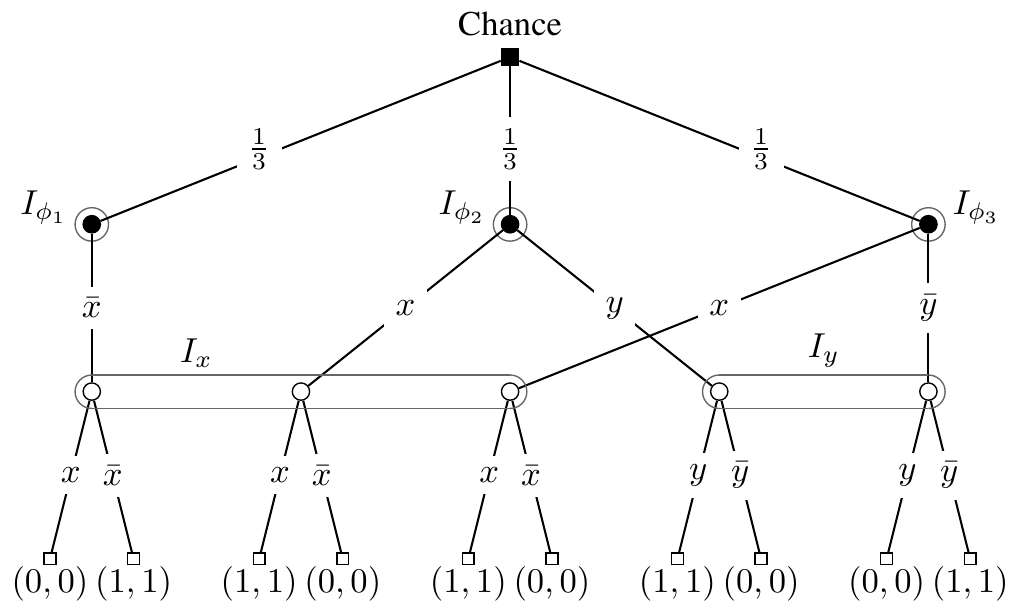}
    \caption{Two-player extensive-form game with Chance built from the SAT instance $C = \{\bar{x}, x \lor y, x \lor \bar{y}\}$.}
    \label{fig:sat_chance_tree}
\end{figure}

\begin{table*}[ht]
    \centering
    \includegraphics[scale=.81]{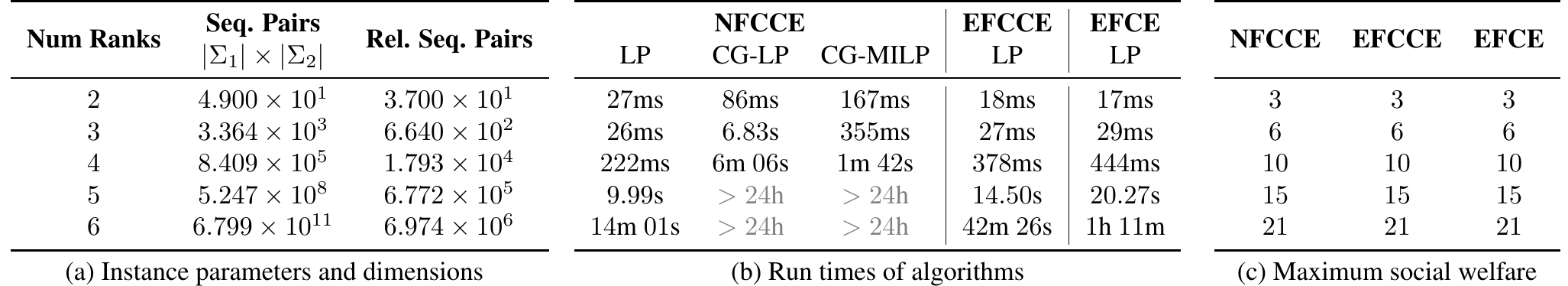}
    \caption{Experimental results on several instances of the Goofspiel game.}
    \label{table:goofspiel}
\end{table*}
    \section{Additional Figures}

    \begin{figure}[H]
        \centering
        \includegraphics[scale=.82]{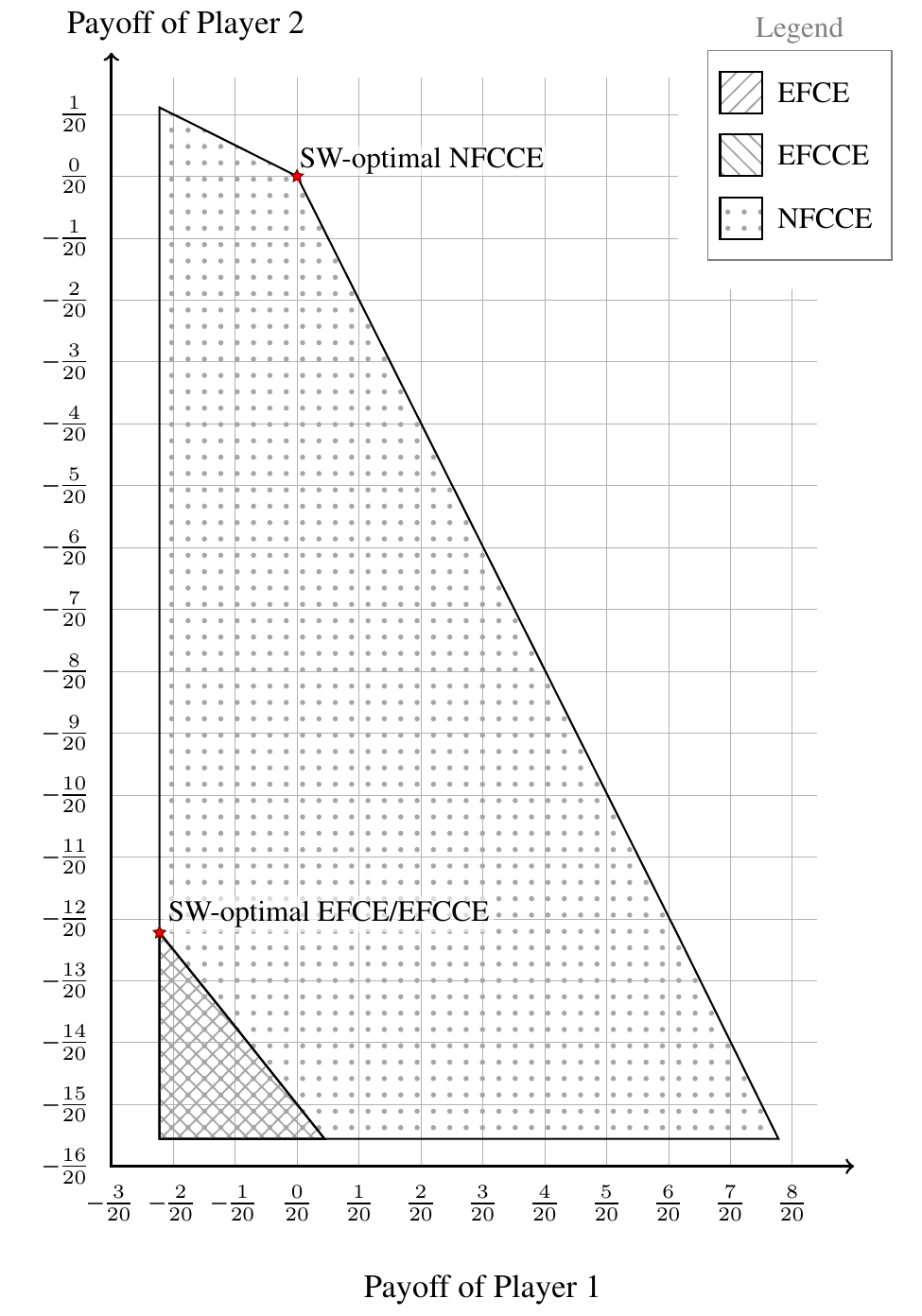}
        \caption{
            Payoff-space representation of a battleship game with a
            $2 \times 1$ board, one ship per player and two rounds.
            In this case,
            $U_\text{EFCE} = U_\text{EFCCE} \subsetneq U_\text{NFCCE}$.
        }
    \end{figure} 

    \section{Game Instances Used in Our Experiments}\label{appendix:exp}

For a detailed explanation of the games of Sheriff and Battleship, see the work by \citet{Farina19:Correlation} that introduced them. In this section, we will only briefly review the most important parameters that characterize them.

Where applicable, we will use the same symbols as~\citet{Farina19:Correlation} in the descriptions of the games below.

\subsection{Sheriff}
The Sheriff instances that we use in our experiments are parametric over the 
maximum number $n_\text{max}$ of illegal items that the Smuggler can load in his or her cargo, the maximum bribe $b_\text{max}$ that can be offered to the Sheriff, and the number of bargaining rounds $r$ between the two players. 
Increasing any of this parameters affects the size of the resulting game instance. The other parameters of the game are set to the fixed values $v = 5, p = 1, s = 1$ in all of our game instances.

\subsection{Battleship}
Our instances of the Battleship game are parametric on the grid size $(w,h)$, and the maximum number of rounds $r$ that players have. The loss multiplier $\gamma$ was set to the fixed value $2$ in all of our instances. Furthermore, each player has exactly one ship of length $1$ and value $1$.

\subsection{Goofspiel}
The variant of Goofspiel~\citep{Ross71:Goofspiel} that we use in our experiments is a two-player card game, employing three identical decks of $r$ cards each. At the beginning of the game, each player receives one of the decks to use it as its own hand, while the last deck is put face down between the players, with cards in increasing order of rank from top to bottom. Cards from this deck will be the prizes of the game. In each round, the players privately select a card from their hand as a bet to win the topmost card in the prize deck. The selected cards are simultaneously revealed, and the highest one wins the prize card. In case of a tie, the prize card is discarded. Each prize card's value is equal to its face value, and at the end of the game the players' score are computed as the sum of the values of the prize cards they have won.

\section{Additional Experimental Results}

See \cref{table:goofspiel} for experimental results on Goofspiel.
 
\fi

\end{document}